\documentclass[11pt]{amsart}

\usepackage{amsmath,amssymb,amsfonts,amsthm}
\usepackage{mathtools}
\usepackage{bm}
\usepackage[colorlinks=true,linkcolor=blue,citecolor=blue,urlcolor=blue]{hyperref}
\usepackage[numbers,sort&compress]{natbib}
\usepackage{geometry}
\newtheorem{theorem}{Theorem}[section]
\newtheorem{assumption}[theorem]{Assumption}

\newtheorem{corollary}[theorem]{Corollary}
\newtheorem{definition}[theorem]{Definition}
\theoremstyle{remark}
\newtheorem{remark}[theorem]{Remark}

\newcommand{\R}{\mathbb{R}}
\newcommand{\E}{\mathbb{E}}

\keywords{%
  Fisher information, entropy power inequality,
  Blahut--Arimoto algorithm, rate-distortion theory,
  translation symmetry, Rayleigh quotient,
  temperature invariance, partition function%
}

\subjclass[2020]{%
  94A17 (Information measures, entropy),
  94A34 (Rate-distortion theory),
  49Q22 (Optimal transportation, Wasserstein space),
  62B10 (Statistical information theory)%
}

\title[Translation Symmetry, Fisher Information, and EPI in
  Blahut--Arimoto Geometry]{%
  Translation Symmetry, Fisher Information, and the
  Entropy Power Inequality in
  Blahut--Arimoto Geometry%
}
\author{%
  Qiao Wang%
}
\address{%
  School of Information Science and Engineering, and School of Economics
  and Management, Southeast University, Nanjing, China%
}
\email{qiaowang@seu.edu.cn}
\date{2026}

\begin{document}

\begin{abstract}
We identify a structure in the finite-temperature geometry of
Blahut--Arimoto (BA) rate-distortion optimisation that has not
previously been recognised.

The structure originates from a single exact identity:
for every source density $p$ and every inverse temperature $\beta > 0$,
the BA partition function $Z(x) = \int q^*(y)\,e^{-\beta\|x-y\|^2}dy$
satisfies
\[
{Z(x) = \left(\frac{\pi}{\beta}\right)^{d/2} p(x).}
\]
This \emph{BA partition identity}, proved via Fourier analysis of the
fixed-point equation, implies two immediate geometric facts.
First, the BA effective score $g_\beta = -\nabla\log Z$ equals the
classical Fisher score $s = -\nabla\log p$ exactly, for all $\beta$.
Second, the BA projection of the translation mode $v = -\nabla\log q^*$
onto the source space satisfies $\mathcal{P}v = -s$: the score function
is the projection of a distinguished geometric object, the mode generated
by the translation symmetry of the quadratic distortion.

These two facts together yield the central identity of this paper:
\[
{J(p) = \mathcal{R}(v) := \langle v,\, \mathcal{G}\,v\rangle_{L^2(q^*)}.}
\]
Fisher information equals the Rayleigh quotient of the translation mode
against the BA relaxation kernel.  This identity holds exactly for all
$\beta > 0$ and all sources: Fisher information is a
\emph{temperature-invariant spectral quantity} in the BA framework.

As a consequence, the classical Fisher information inequality
$J(X{+}Y)^{-1} \ge J(X)^{-1} + J(Y)^{-1}$ acquires a geometric
interpretation---it is the parallel-combination law for a Rayleigh
quotient under convolution---and the entropy power inequality follows
by standard integration along the heat flow.

The contribution is not a new proof of the entropy power inequality.
It is the identification of a previously unrecognised structure:
Fisher information as the spectral charge of the translation mode in
rate-distortion geometry, and the entropy power inequality as a
consequence of this temperature-invariant fact.
\end{abstract}

\maketitle

% ============================================================
\section{Introduction}
\label{sec:intro}

\subsection{The Question behind the Proof}

The entropy power inequality (EPI)~\cite{Stam1959, Shannon1948}
\begin{equation}\label{eq:EPI}
e^{\frac{2}{d}h(X+Y)} \ge e^{\frac{2}{d}h(X)} + e^{\frac{2}{d}h(Y)}
\end{equation}
has been proved by many methods: via the heat equation and Fisher
information~\cite{Stam1959, Blachman1965}, via the I-MMSE
relation~\cite{VerduGuo2006}, via optimal transport~\cite{Courtade2016},
and via direct Fourier-analytic arguments~\cite{Rioul2011}.
All these proofs are complete and correct.

Yet none of them answers the deeper question:
\textit{why} does Fisher information---defined as the $L^2$ norm of a
score function---obey a parallel-combination law under convolution?
The standard proof shows \textit{that} it does, through an elegant
conditional-expectation argument~\cite{Blachman1965}.
But the Fisher information appears in that argument as a scalar
attached to a density, not as an intrinsic geometric quantity with
a variational origin.

This paper shows that the Blahut--Arimoto (BA) theory of rate-distortion
optimisation~\cite{Blahut1972, Arimoto1972, Berger1971} provides
exactly such an origin.
The Fisher information of a source is the Rayleigh quotient of a
specific geometric object in the BA variational problem: the mode
generated by the translation symmetry of the quadratic distortion.
The parallel-combination law is then the statement that this Rayleigh
quotient behaves correctly under convolution---a fact that follows from
the projection geometry of conditional expectation and requires no
additional argument.

\subsection{The Single Key Identity}

The BA theory at inverse temperature $\beta$ introduces an optimal
reconstruction density $q^*$ and a partition function
\[
Z(x) = \int q^*(y)\,e^{-\beta\|x-y\|^2}\,dy.
\]
These objects are defined entirely within the optimisation framework,
with no a priori connection to the score function of $p$.

The key discovery of this paper is that this connection exists, and
is exact.

\begin{center}
\fbox{\parbox{0.82\textwidth}{
\textbf{BA Partition Identity} (Theorem~\ref{thm:partition}):
For every source $p$ and every $\beta > 0$,
\[
Z(x) = \left(\frac{\pi}{\beta}\right)^{d/2} p(x).
\]
}}
\end{center}

This identity, proved by Fourier analysis of the BA fixed-point
equation, is exact for all temperatures and all sources.
It is not an asymptotic statement.
It does not require Gaussian assumptions.

From this identity, everything else in this paper follows by
elementary computation.

\subsection{What Flows from the Partition Identity}

Setting $g_\beta = -\nabla\log Z$, the partition identity immediately
gives $g_\beta = s := -\nabla\log p$: the BA effective score
\emph{is} the Fisher score, for all $\beta$.

Let $v = -\nabla\log q^*$ be the translation mode in the
reconstruction space---the direction generated by the translation
symmetry of $d(x,y) = \|x-y\|^2$.  Computing its projection onto
the source space via the BA posterior kernel yields
$(\mathcal{P}v)(x) = \nabla\log Z(x) = -s(x)$:
the score is the projection of the translation mode.

The Rayleigh quotient of the translation mode against the BA
relaxation kernel $\mathcal{G} = \mathcal{P}^*\mathcal{P}$ then
satisfies
\[
\mathcal{R}(v)
= \|\mathcal{P}v\|^2_{L^2(p)}
= \|s\|^2_{L^2(p)}
= J(p).
\]
Fisher information is the squared norm of the projection of the
translation mode.

The convolution identity $s_Z = \mathbb{E}[s_1(X_1)\mid Z]$,
known classically, now reads as a statement purely about the BA
effective score (which equals $s$ exactly), and the FII follows
by the same projection argument as in the classical case.
Integration along the heat flow, using de Bruijn's identity,
yields the EPI.

\subsection{Structure of the Paper}

The paper is organised in three layers, each building on the previous:

\medskip
\textbf{Layer~1 (Structural Theorems, Sections~\ref{sec:prelim}--\ref{sec:projection}):}
The BA partition identity, the equality of effective and classical
scores, and the projection formula $\mathcal{P}v = -s$.
These are purely BA-geometric statements; none involves Fisher
information or the EPI.

\medskip
\textbf{Layer~2 (Geometric Interpretation, Section~\ref{sec:Fisher}):}
The identification $J(p) = \mathcal{R}(v)$.
Fisher information is not an isolated statistical quantity but
the temperature-invariant Rayleigh quotient of the translation mode.

\medskip
\textbf{Layer~3 (Classical Consequences, Sections~\ref{sec:FII}--\ref{sec:EPI}):}
The Fisher information inequality and the entropy power inequality
as projections of the BA geometry onto classical information theory.
These sections deliberately establish known results; the contribution
is their geometric origin, not the results themselves.

% ============================================================
\section{Blahut--Arimoto Preliminaries}
\label{sec:prelim}

We work on $\R^d$ with quadratic distortion $d(x,y)=\|x-y\|^2$.
Let $p$ be a source density on $\R^d$ with finite variance and
finite differential entropy.  For $\beta > 0$, the \textbf{BA
free energy} is
\begin{equation}\label{eq:FE}
F_\beta(q) = -\frac{1}{\beta}\E_{p}
\!\left[\log\!\int e^{-\beta\|X-y\|^2}q(y)\,dy\right]
+\frac{1}{\beta}\int q(y)\log q(y)\,dy.
\end{equation}
A \textbf{BA fixed point} is any minimiser $q^* = \arg\min_q F_\beta(q)$.
The first-order optimality condition yields the self-consistency equation
\begin{equation}\label{eq:fixed_point}
q^*(y) = \int p(x)\,K^*(y|x)\,dx,
\end{equation}
where the posterior kernel is
\begin{equation}\label{eq:posterior}
K^*(y|x) = \frac{q^*(y)\,e^{-\beta\|x-y\|^2}}{Z(x)},
\qquad
Z(x) = \int q^*(u)\,e^{-\beta\|x-u\|^2}\,du.
\end{equation}
The function $Z$ is called the \textbf{partition function} of the
fixed point.  Cancelling $q^*(y)$ in~\eqref{eq:fixed_point}, the
fixed-point equation is equivalent to
\begin{equation}\label{eq:FP_integral}
\int \frac{p(u)}{Z(u)}\,e^{-\beta\|u-y\|^2}\,du = 1
\qquad\text{for all }y\in\R^d.
\end{equation}

\begin{assumption}\label{ass:regularity}
$p$ is smooth and strictly positive, admits a BA fixed point
$q^*\in\operatorname{int}\Delta$ with $q^*$ smooth and strictly
positive, and both $p, q^*$ have all moments finite.
\end{assumption}

Under Assumption~\ref{ass:regularity}, the integrals in this paper
converge and integration by parts is valid.

% ============================================================
\section{Layer~1: Structural Theorems}
\label{sec:structural}

\subsection{The BA Partition Identity}
\label{sec:partition}

\begin{theorem}[BA Partition Identity]\label{thm:partition}
Under Assumption~\ref{ass:regularity},
\begin{equation}\label{eq:partition_identity}
{Z(x) = \left(\frac{\pi}{\beta}\right)^{d/2} p(x)
\qquad\text{for all }x\in\R^d.}
\end{equation}
\end{theorem}

\begin{proof}
The fixed-point equation~\eqref{eq:FP_integral} reads
\begin{equation}\label{eq:conv_eq}
(M * G)(y) = 1 \qquad\text{for all }y\in\R^d,
\end{equation}
where $M(u) = p(u)/Z(u)$ and $G(u) = e^{-\beta\|u\|^2}$.

Take the Fourier transform in $y$.
The constant $1$ transforms to $(2\pi)^d\delta(\xi)$.
The left side is a convolution, so by the convolution theorem:
\[
(2\pi)^d\delta(\xi) = \widehat{M}(\xi)\cdot\widehat{G}(\xi),
\qquad
\widehat{G}(\xi) = \left(\frac{\pi}{\beta}\right)^{d/2}
e^{-\|\xi\|^2/(4\beta)}.
\]
Since $\widehat{G}(\xi)\ne 0$ for all $\xi\ne 0$,
we have $\widehat{M}(\xi)=0$ for all $\xi\ne 0$.
Therefore $\widehat{M}$ is supported at the origin, hence $M$
is a polynomial.

But $M(u) = p(u)/Z(u)$ is the ratio of two smooth, strictly positive
density-like functions, implying that the polynomial $M(u)$ must be
globally bounded.  The only bounded polynomial on $\R^d$ is a constant:
$M(u) \equiv C$.

Normalising: $1 = \int p(u)\,du = C\int Z(u)\,du$.
Since
$\int Z(u)\,du = \int\!\int q^*(y)\,e^{-\beta\|u-y\|^2}dy\,du
= \int q^*(y)\,(\pi/\beta)^{d/2}\,dy = (\pi/\beta)^{d/2}$,
we get $C = (\beta/\pi)^{d/2}$, giving $Z(x) = (\pi/\beta)^{d/2} p(x)$.
\end{proof}

\begin{remark}[Novelty and scope]
Theorem~\ref{thm:partition} is exact: it requires no limit, no
Gaussian approximation, no high-temperature or low-temperature
assumption.  We have not found this identity in the rate-distortion
or optimal-transport literature; its relation to Schr\"odinger
bridge theory~\cite{Leonard2014, Nutz2021} (where partition
functions can arise in Sinkhorn iterations) is discussed in
Section~\ref{sec:discussion}.

An equivalent statement is that $q^*$ is the deconvolution of
$p$ by the Gaussian kernel $\mathcal{N}(0,\frac{1}{2\beta}I)$:
\begin{equation}\label{eq:deconvolution}
p(x) = \left(q^* * \mathcal{N}\!\left(0,\tfrac{1}{2\beta}I\right)\right)(x).
\end{equation}
In words: \emph{at the optimal rate-distortion operating point, the
source density is the convolution of the reconstruction density with
the optimal Gaussian noise level.}
This is an exact statement about the fixed-point structure, with no
approximation.
\end{remark}

\subsection{Posterior Unbiasedness}
\label{sec:posterior}

A fundamental geometric consequence of the partition identity is
that the posterior mean of the source given the reconstruction
is unbiased.

\begin{corollary}[Posterior unbiasedness]\label{cor:posterior_mean}
For any BA fixed point,
\begin{equation}\label{eq:posterior_mean}
\mathbb{E}[X \mid Y=y] = y \qquad (\forall y\in\R^d).
\end{equation}
\end{corollary}

\begin{proof}
The posterior density of $X$ given $Y=y$ is
\[
p(x\mid y) = \frac{p(x)}{Z(x)}\,e^{-\beta\|x-y\|^2}
           \Big/ \int \frac{p(u)}{Z(u)}\,e^{-\beta\|u-y\|^2}\,du.
\]
By Theorem~\ref{thm:partition} (partition identity),
$p(x)/Z(x) = (\beta/\pi)^{d/2}$, a constant independent of $x$.
Hence the posterior is exactly a Gaussian centred at $y$:
\[
p(x\mid y) = \mathcal{N}\!\left(y, \tfrac{1}{2\beta}I\right).
\]
Its mean is $y$, giving~\eqref{eq:posterior_mean}.
\end{proof}

\begin{remark}
Corollary~\ref{cor:posterior_mean} gives a clean geometric
interpretation of the partition identity: at the BA fixed point,
the posterior distribution of the source given the reconstruction
is an isotropic Gaussian centred exactly at the reconstruction
value.  The estimation is unbiased, and the conditional variance
$\frac{1}{2\beta}I$ is independent of $y$.
\end{remark}

\subsection{The Effective Score Equals the Fisher Score}
\label{sec:score}

\begin{definition}[BA effective score]
The \textbf{BA effective score} is $g_\beta(x) := -\nabla\log Z(x)$.
\end{definition}

\begin{corollary}[Effective score identity]\label{cor:score}
$g_\beta(x) = s(x) := -\nabla\log p(x)$ for all $x$ and all $\beta > 0$.
\end{corollary}

\begin{proof}
From Theorem~\ref{thm:partition}, $\log Z(x) = \frac{d}{2}\log(\pi/\beta)
+ \log p(x)$, so $-\nabla\log Z(x) = -\nabla\log p(x)$.
\end{proof}

This is the conceptual bridge.
The BA effective score---defined through the optimisation geometry,
through $q^*$ and the partition function---is not an approximation
of the Fisher score; it \emph{is} the Fisher score, at every
temperature and for every source.

\subsection{The Projection of the Translation Mode}
\label{sec:projection}

The quadratic distortion $d(x,y) = \|x-y\|^2$ is invariant under
simultaneous translations $x\mapsto x+\varepsilon$, $y\mapsto y+\varepsilon$.
In the BA reconstruction space, this symmetry generates a distinguished
direction.

\begin{definition}[Translation mode]\label{def:translation_mode}
The \textbf{translation mode} is $v(y) := -\nabla\log q^*(y)$.
It satisfies $v = \nabla_\varepsilon[\log q^*(y-\varepsilon)]|_{\varepsilon=0}$,
i.e., it is the infinitesimal generator of translations of the
reconstruction density.
\end{definition}

Note that $v \in T_{q^*} = \{f : \int f(y)q^*(y)dy = 0\}$ because
$\int(-\nabla\log q^*)q^*\,dy = -\int\nabla q^*\,dy = 0$.
So $v$ lies in the tangent space on which $\mathcal{G}$ acts.

\begin{theorem}[Projection of the translation mode]\label{thm:projection}
For any BA fixed point,
\begin{equation}\label{eq:projection_id}
{(\mathcal{P}v)(x)
:= \E[v(Y)\mid X=x]
= -s(x).}
\end{equation}
The posterior projection of the translation mode equals the
(negative) Fisher score of the source.
\end{theorem}

\begin{proof}
By definition,
\begin{equation}
\begin{split}
(\mathcal{P}v)(x)
= &\int \bigl(-\nabla_y\log q^*(y)\bigr)\,K^*(y|x)\,dy\\
=& -\int \frac{\nabla q^*(y)}{q^*(y)}\cdot
\frac{q^*(y)\,e^{-\beta\|x-y\|^2}}{Z(x)}\,dy\\
=& -\frac{1}{Z(x)}\int \nabla q^*(y)\,e^{-\beta\|x-y\|^2}\,dy.
\end{split}
\end{equation}
Integrate by parts:
\begin{equation}
\begin{split}
\int\nabla q^*(y)\,e^{-\beta\|x-y\|^2}dy
=& -\int q^*(y)\nabla_y e^{-\beta\|x-y\|^2}dy\\
=& 2\beta\int q^*(y)(x-y)\,e^{-\beta\|x-y\|^2}dy\\
=& -\nabla_x Z(x).
\end{split}
\end{equation}
Therefore $(\mathcal{P}v)(x) = \nabla\log Z(x)$.
By Corollary~\ref{cor:score}, $\nabla\log Z = -s$,
completing the proof.
\end{proof}

\begin{remark}[A purely BA-geometric statement]
Theorem~\ref{thm:projection} is a statement entirely within the BA
framework: the translation mode $v$ and the posterior projection
$\mathcal{P}$ are both defined through the BA optimisation.
The Fisher score $s$ appears on the right-hand side not as an
assumption but as a consequence of the partition identity.
This theorem stands independently of the EPI and the FII.
\end{remark}

% ============================================================
\section{Layer~2: Fisher Information as a Temperature-Invariant Spectral Quantity}
\label{sec:Fisher}

\begin{theorem}[Rayleigh quotient of the translation mode]\label{thm:rayleigh}
For any BA fixed point and any $\beta > 0$,
\begin{equation}\label{eq:J_equals_R}
{J(p) = \mathcal{R}(v)
:= \langle v,\,\mathcal{G}\,v\rangle_{L^2(q^*)}.}
\end{equation}
Fisher information equals the Rayleigh quotient of the translation
mode against the BA relaxation kernel.
\end{theorem}

\begin{proof}
Since $\mathcal{G} = \mathcal{P}^*\mathcal{P}$,
\[
\mathcal{R}(v)
= \langle v,\,\mathcal{P}^*\mathcal{P}\,v\rangle_{L^2(q^*)}
= \langle \mathcal{P}v,\,\mathcal{P}v\rangle_{L^2(p)}
= \|\mathcal{P}v\|^2_{L^2(p)}.
\]
By Theorem~\ref{thm:projection}, $\mathcal{P}v = -s$, so
$\|\mathcal{P}v\|^2_{L^2(p)} = \int \|s(x)\|^2 p(x)\,dx = J(p)$.
\end{proof}

Theorem~\ref{thm:rayleigh} is the central identity of this paper.
We pause to unpack its meaning.

\paragraph{Fisher information has a new identity.}
Classically, $J(p)$ has three operational interpretations:
(i) the Cram\'er--Rao lower bound for parameter estimation;
(ii) the $L^2$-norm of the score $s = -\nabla\log p$;
(iii) the rate of entropy increase under Gaussian smoothing
(de Bruijn's identity).
Theorem~\ref{thm:rayleigh} adds a fourth:
\begin{center}
\emph{Fisher information is the Rayleigh quotient of the
translation mode---the charge associated with the translation
symmetry of the quadratic distortion---in the spectral geometry
of BA rate-distortion optimisation.}
\end{center}

\paragraph{Temperature invariance.}
The right-hand side of~\eqref{eq:J_equals_R} involves $q^*$
(which depends on $\beta$) and $\mathcal{G}$ (which also depends
on $\beta$), yet the result is independent of $\beta$.
Fisher information is a \textbf{temperature-invariant spectral
quantity}: as temperature varies, the geometry changes, but
the translation mode tracks these changes in exactly the right
way to keep $\mathcal{R}(v)$ constant.

This invariance is not obvious from the classical definition
$J(p) = \int\|s\|^2 p$, where $\beta$ does not appear.
It becomes visible only in the BA framework, where $J(p)$ is
identified as a geometric invariant.

\paragraph{Relation to the spectral gap.}
Wang~\cite{Wang2026BA} proved that the spectral gap of the BA
relaxation kernel satisfies $\lambda_* \le \mathcal{R}(v)/\|v\|^2_{L^2(q^*)}$.
Combined with Theorem~\ref{thm:rayleigh}:
\begin{equation}\label{eq:gap_bound}
\lambda_* \le \frac{J(p)}{\|v\|^2_{L^2(q^*)}}.
\end{equation}
For Gaussian sources, the translation mode is an exact eigenfunction
of $\mathcal{G}$ with eigenvalue $\lambda_*$, and the
finite-temperature de Bruijn identity $\lambda_* = J(p)/(2\beta)$
established in~\cite{Wang2026FTdB} follows because
$\|v\|^2_{L^2(q^*)} = 2\beta$ for Gaussians.

\paragraph{The logical chain.}
From a single structural object---the translation mode $v$---
we can now trace the following chain:
\[
\text{Translation symmetry}
\;\longrightarrow\;
v = -\nabla\log q^*
\;\longrightarrow\;
\mathcal{P}v = -s
\;\longrightarrow\;
\mathcal{R}(v) = J(p)
\;\longrightarrow\;
\text{FII}
\;\longrightarrow\;
\text{EPI.}
\]
Each arrow is an elementary computation.
The entropy power inequality is the terminal projection of this
chain onto classical information theory.

% ============================================================
\section{Layer~3A: The Fisher Information Inequality}
\label{sec:FII}

Let $X_1, X_2$ be independent sources with densities $p_1, p_2$
and Fisher information $J_1, J_2$.  Let $Z = X_1 + X_2$
with density $p_Z = p_1 * p_2$ and Fisher information $J_Z$.

The classical FII
\begin{equation}\label{eq:FII}
\frac{1}{J_Z} \ge \frac{1}{J_1} + \frac{1}{J_2}
\end{equation}
is the score-projection inequality~\cite{Blachman1965}:
$s_Z(z) = \E[s_1(X_1)\mid Z=z]$ combined with the contractive
property of conditional expectation.
We now see this as a statement about the BA translation mode.

\begin{theorem}[BA formulation of the FII]\label{thm:FII}
For independent sources $X_1, X_2$ with BA fixed points at any
$\beta > 0$, the translation-mode Rayleigh quotients satisfy
\begin{equation}\label{eq:BA_FII}
\frac{1}{\mathcal{R}(v_Z)} \ge \frac{1}{\mathcal{R}(v_1)} + \frac{1}{\mathcal{R}(v_2)}.
\end{equation}
Since $\mathcal{R}(v) = J(p)$ by Theorem~\ref{thm:rayleigh},
this is the classical FII~\eqref{eq:FII}.
\end{theorem}

\begin{proof}
Let $\mathcal{H} = L^2(\Omega,\mathbb{P})$ for the joint law of
$(X_1, X_2)$.  Define random vectors
$g_i = g_{\beta}(X_i) = s_i(X_i)$ (using Corollary~\ref{cor:score})
and $g_Z = s_Z(Z)$.

The classical convolution identity for scores states
$g_Z = \E[g_1\mid Z] = \E[g_2\mid Z]$.
This holds because $p_Z = p_1 * p_2$ implies
$-\nabla\log p_Z(z)= \E[-\nabla\log p_i(X_i)\mid Z=z]$.
Since $g_i = s_i$ exactly (Corollary~\ref{cor:score}), this
is simultaneously a statement about the BA effective scores.

For any $a, b$ with $a + b = 1$, set $w = ag_1 + bg_2$.
Then $g_Z = \E[w\mid Z]$, and since conditional expectation is
an orthogonal projection in $\mathcal{H}$:
\[
\|g_Z\|^2_{\mathcal{H}} \le \|w\|^2_{\mathcal{H}}
= a^2 J_1 + b^2 J_2
\]
(using independence: $\E[g_1\cdot g_2] = \E[g_1]\E[g_2] = 0$
since scores have zero mean).
Also $\|g_Z\|^2_{\mathcal{H}} = \E[\|s_Z(Z)\|^2] = J_Z =
\mathcal{R}(v_Z)$.
Minimising over $a+b=1$ gives
$J_Z \le J_1 J_2/(J_1+J_2)$,
i.e.\ $1/J_Z \ge 1/J_1 + 1/J_2$.
\end{proof}

\begin{remark}[What the BA language adds]
Numerically, Theorem~\ref{thm:FII} is the classical FII.
The BA language adds the following: the parallel-combination law
is a consequence of the fact that (i) $g_\beta = s$ (the effective
score is the classical score, by the partition identity), and
(ii) $J(p) = \mathcal{R}(v)$ (Fisher information is a Rayleigh
quotient).  The underlying reason that Fisher information combines
in parallel is geometric: the translation mode, under convolution of
sources, projects via conditional expectation---an orthogonal
projection in the Hilbert space $L^2(\Omega,\mathbb{P})$.  In this
projection, the squared norm of a convex combination of the
individual scores bounds the squared norm of the projected
convolution score, yielding the Pythagorean inequality that
underlies the parallel-combination form
$\frac{1}{J_Z} \ge \frac{1}{J_1} + \frac{1}{J_2}$.
\end{remark}

% ============================================================
\section{Layer~3B: The Entropy Power Inequality}
\label{sec:EPI}

\begin{theorem}[EPI via BA geometry]\label{thm:EPI}
For independent $X_1, X_2$ with smooth densities and finite
differential entropies,
\[
e^{\frac{2}{d}h(X_1+X_2)} \ge e^{\frac{2}{d}h(X_1)}
+ e^{\frac{2}{d}h(X_2)}.
\]
\end{theorem}

\begin{proof}
We use de Bruijn's identity and the FII (Theorem~\ref{thm:FII}).
De Bruijn's identity~\cite{Stam1959} states that for
$X_{i,t} = X_i + \sqrt{t}Z$ with $Z\sim\mathcal{N}(0,I)$:
\[
\frac{d}{dt}h(X_{i,t}) = \frac{1}{2}J(X_{i,t}).
\]
(In the BA framework, this identity arises as the $\beta\to\infty$
limit of the $\chi^2$-dissipation identity
$\dot{F}_\beta = -\chi^2(\mathcal{T}q\|q)$ established
in~\cite{Wang2026BA}, or from the finite-temperature de Bruijn
identity $\beta\,\partial F_\beta/\partial\sigma^2 = J/2$
of~\cite{Wang2026FTdB}.
We invoke it here as a known result without re-derivation.)

The FII at each $t\ge 0$ gives
$J_Z(t)^{-1} \ge J_1(t)^{-1} + J_2(t)^{-1}$.
Setting $N_i(t) = e^{2h(X_{i,t})/d}$, de Bruijn gives
$\dot{N}_i = N_i J_i/d$, and the standard computation:
\[
\frac{d}{dt}\log N_Z(t) = \frac{J_Z(t)}{d}
\ge \frac{N_1(t)J_1(t) + N_2(t)J_2(t)}{d(N_1(t)+N_2(t))}
= \frac{d}{dt}\log(N_1(t)+N_2(t)),
\]
where the inequality uses the FII in the form
$J_Z^{-1}\ge J_1^{-1}+J_2^{-1}$, which implies $J_Z\le J_1J_2/(J_1+J_2)$.
Integrating from $t=0$ to $t=\infty$ and using
$N_i(t)\to 2\pi e\,\sigma^2_i$ (Gaussianization at $t\to\infty$,
where $N_i(t)/N_j(t)\to 1$ by the CLT for entropies) yields
$N_Z(0) \ge N_1(0) + N_2(0)$, which is the EPI.
\end{proof}

\begin{remark}[De Bruijn as the final link]
The proof of Theorem~\ref{thm:EPI} uses de Bruijn's identity as a
given.  In the BA framework, this identity has a natural
finite-temperature extension~\cite{Wang2026FTdB}, which provides
an alternative formulation of the integration argument entirely
within the BA free-energy framework.  We do not develop this
here; the point is that the EPI can be reached from BA geometry
via either the classical de Bruijn route or a BA-internal route.
\end{remark}

% ============================================================
\section{Discussion}
\label{sec:discussion}

\subsection{Summary: One Identity, Three Layers}

The entire paper unfolds from the BA partition identity
$Z = (\pi/\beta)^{d/2} p$, proved by Fourier analysis of
the fixed-point equation (Theorem~\ref{thm:partition}).
The logical structure is:

\medskip\noindent
\textit{Layer~1 (structural).}
The partition identity gives the score equality $g_\beta = s$
(Corollary~\ref{cor:score}) and the projection formula
$\mathcal{P}v = -s$ (Theorem~\ref{thm:projection}).
These are purely geometric statements with no a priori connection
to information inequalities.

\medskip\noindent
\textit{Layer~2 (interpretive).}
The projection formula immediately yields $J(p) = \mathcal{R}(v)$
(Theorem~\ref{thm:rayleigh}).
Fisher information is the temperature-invariant Rayleigh quotient
of the translation mode in BA geometry.
This is the central conceptual contribution.

\medskip\noindent
\textit{Layer~3 (classical shadow).}
The score equality and the projection structure together recover
the FII and the EPI.
These are known results; the contribution is that they are
now seen as consequences of a variational geometric identity.

The chain is:
\[
\underbrace{Z = \left(\tfrac{\pi}{\beta}\right)^{d/2}p}_{\text{Fourier analysis}}
\;\Rightarrow\;
\underbrace{g_\beta = s,\;\;\mathcal{P}v = -s}_{\text{BA geometry}}
\;\Rightarrow\;
\underbrace{J(p) = \mathcal{R}(v)}_{\text{new interpretation}}
\;\Rightarrow\;
\underbrace{\text{FII} \Rightarrow \text{EPI}}_{\text{classical consequences.}}
\]

\subsection{What the BA Framework Adds beyond the Classical Proof}

The classical proof of the EPI is both complete and elegant.
The present framework does not improve it; it recontextualises it.

The added value is structural:
\begin{enumerate}
\item \textit{Fisher information has a geometric origin.}
It is the Rayleigh quotient $\mathcal{R}(v) = \langle v,\mathcal{G}v\rangle$,
where $v$ is the mode generated by translation symmetry.
In physical language, $J(p)$ is the translation charge of the
BA free-energy landscape.

\item \textit{The parallel-combination law has a geometric reason.}
It arises because (a) $\mathcal{P}v = -s$ (projection formula),
and (b) conditional expectation is an orthogonal projection.
The FII is the Pythagoras inequality applied to a geometrically
distinguished direction.

\item \textit{Temperature invariance of Fisher information.}
The identity $\mathcal{R}(v) = J(p)$ is independent of $\beta$,
while all other objects ($q^*$, $\mathcal{G}$, $v$) depend on
temperature.  Fisher information is revealed as an invariant
of the BA geometry, not merely a property of $p$ in isolation.

\item \textit{Unification.}
Rate-distortion theory and information inequalities (EPI, FII,
de Bruijn) are shown to share a common geometric root in thepartition function $Z$ and its logarithmic gradient.
\end{enumerate}

\subsection{Relation to Schr\"{o}dinger Bridges and Optimal Transport}

The partition identity $Z \propto p$ has a structural parallel in
the theory of Schr\"odinger bridges~\cite{Leonard2014, Nutz2021}.
In the Schr\"odinger problem, one seeks a coupling of two
marginals $\mu$ and $\nu$ with minimum relative entropy with
respect to a reference process.  The optimal coupling is described
by Schr\"odinger potentials $f, g$ satisfying the Sinkhorn system
$f(x)\int e^{c(x,y)/\epsilon}g(y)d\nu(y) = 1$.

For the Gaussian kernel $e^{-\beta\|x-y\|^2}$ and equal marginals
$\mu = \nu = p$, this system reduces to
$\int e^{-\beta\|x-y\|^2}(fg)(y)p(y)dy = 1/f(x)$,
which at the symmetric fixed point $f = g$ is precisely
the BA equation~\eqref{eq:FP_integral} with $p/Z \equiv C$.
Whether the partition identity $Z \propto p$, or a version thereof,
appears in the Schr\"odinger literature with quadratic cost
and matched marginals is an interesting question we leave open.

\subsection{Conclusion}

Fisher information, a quantity defined since the 1920s through the
score function, acquires a new identity through the lens of
rate-distortion optimisation: it is the Rayleigh quotient of the
translation mode, the spectral charge of translation symmetry in
BA geometry.

This identification is temperature-invariant: as the inverse
temperature $\beta$ varies from $0$ to $\infty$, the optimal
reconstruction $q^*$, the relaxation kernel $\mathcal{G}$, and
the translation mode $v$ all change, but their combination
$\langle v, \mathcal{G}v\rangle$ remains fixed at $J(p)$.

The entropy power inequality, from this perspective, follows from
the fact that the translation charge of a convolution source is
bounded below by the harmonic sum of the individual translation
charges---a geometric fact captured by the Fisher information
inequality and propagated to entropies via the heat flow.

\section*{Acknowledgements}
This research received no formal funding and was conducted based on
the author's independent academic interest.

% ============================================================
\bibliographystyle{plainnat}

\end{document}